\documentclass[10pt,tightenlines,eqsecnum,floats,aps,amsmath,
amssymb,nofootinbib,prd,noshowpacs,notitlepage,longbibliography]{revtex4-1}
\usepackage[utf8]{inputenc}
\usepackage{amsmath,amsthm,latexsym,amssymb,amsfonts,color}

\usepackage{hyperref}

\newcommand{\R}{{\mathbb R}}
\DeclareMathOperator{\KG}{KG}
\DeclareMathOperator{\sign}{sign}
\DeclareMathOperator{\sgn}{sgn}
\DeclareMathOperator{\End}{End}

\newtheorem{lm}{Lemma}
\newtheorem{thm}{Theorem}

\begin{document}

\title{Elementary proof of symmetry of the off-diagonal Seeley-DeWitt (and related Hadamard) coefficients.}
\author{Wojciech Kamiński}
\affiliation{Instytut Fizyki Teoretycznej, Wydzia{\l} Fizyki, Uniwersytet Warszawski, ul. Pasteura 5 PL-02093 Warszawa, Poland}

\begin{abstract}
We will prove in an elementary way that off-diagonal Seeley-DeWitt and Hadamard coefficients are (sesqui-)symmetric for smooth manifolds of arbitrary signature.
\end{abstract}

\maketitle

\section{Introduction}

Let us consider a smooth $d$ dimensional manifold $\mathcal{M}$ with the metric $g_{\mu\nu}$\footnote{We assume that manifolds are Hausdorff. A Hausdorff manifold equipped with a metric is automatically paracompact \cite{Geroch}.} and the field $\phi$ valued in the bundle $\mathcal{V}$ associated with $SO$ (or possibly $Spin$) structure, $\phi\in \Gamma({\mathcal M},\mathcal{V})$.  We are interested in the operator
\begin{equation}\label{eq-KG}
 \KG \phi=\left(\nabla_\mu+A_\mu\right)\left(\nabla^\mu+A^\mu\right)\phi
+B\phi
\end{equation}
that is a generalization of the Klein-Gordon operator. Here $A^\mu$ is a smooth matrix valued vector field and $B$ is a smooth matrix valued function
\begin{equation}
A^\mu(x)\colon \mathcal{V}_x\rightarrow \mathcal{V}_x,\quad B(x)\colon \mathcal{V}_x\rightarrow \mathcal{V}_x.
\end{equation}
We assume that $\mathcal V$ carries a $SO$ (or $Spin$) invariant sesqui-linear form (nondegenerate, but not necessarily positively definite), so we can define adjoint operation $\dagger$\footnote{In the case of real field one can consider complexification of the bundle. In this case we assume existence of a symmetric form on original $\mathcal V$ and construct from it a sesqui-linear form on $\mathcal V^{\mathbb C}$.}.
We assume that $(A^\mu)^\dagger=-A^\mu$ and $B^\dagger=B$ so the operator $\KG$ is formally hermitian. 
If $\mathcal V$ carries additional gauge structure, $A^\mu$ should incorporate also contributions from suitable connection.

The original Klein-Gordon operator for the scalar function (on a Lorentzian manifold in dimension $d$ with signature $-+\cdots +$) 
\begin{equation}
 \KG \phi=\phi_{;\mu}^{;\mu}-(M^2+\xi R)\phi
\end{equation}
belongs to the class with $A^\mu=0$ and $B=-M^2-\xi R$. 

We are interested in the (generalized wave) equation $\KG \phi=0$, its various parametrices (approximate Green functions) and distinguished approximate solutions. There is a costruction of the asymptotic expansion of the kernels of Hadamard parametrices (advanced and retarded propagators in the Lorentzian setup) as well as in asymptotic expansion of the Hadamard states (defined for $\KG$ scalar field on Lorentzian manifold, see \cite{Fulling}). In any case the formula is given (in some small neighbourhood of diagonal of $\mathcal{M}\times \mathcal{M}$) by an asymptotic series (see \cite{Fulling, Baer})
\begin{equation}\label{expansion}
\frac{1}{(4\pi)^{\frac{d}{2}}} \sum_{n=0} f_n^{SD}(x,x')R_n(x,x'),
\end{equation}
where the so called Seeley-DeWitt coefficients $f_n^{SD}$ are smooth functions taking values in linear operators from fiber in point $x'$ to fiber in point $x$:
\begin{equation}
f_n^{SD}(x,x')\colon \mathcal{V}_{x'}\rightarrow \mathcal{V}_{x}.
\end{equation}
The distributions $R_n$ are versions of scalar Riesz functions (different in each case above, that are functions of the half of the signed square geodesic distance $\sigma(x,x')$ and an additional data of causal relation of the points $x$ and $x'$) \footnote{For notation convention see section \ref{sec:notation}.}. Let us notice that we use slightly nonstandard normalization of the Riesz distributions, otherwise $f_n^{SD}$ should be rescaled to the Hadamard coefficients (see \cite{Hack2012} for thorought discussion). The rescaling constant depends only on the dimenion $d$ of the manifold  and $n$ (see \cite{Hack2012}).

The function $\sigma(x,x')$ is a covariant quantity, but in general it is well-defined only for sufficiently close points. Functions $f_n^{SD}$ depend on sufficiently high jets of the metric and coefficients of $A^\mu$ and $B$ along the geodesic connecting points $x$ and $x'$. They are also covariantly defined.

The coefficients $f_n^{SD}$ and closely related Hadamard and DeWitt-Schwinger functions (that are rescaled coefficients $f_n^{SD}$ and linear combination of them respectively, see \cite{Hack2012}) are important in application to quantum field theory in curved spacetime where they appear in effective actions as well as in regularization of the stress energy tensor. Indeed, the important property used in the stress energy tensor renormalization is the symmetry of the $f_n^{SD}$ coefficients \cite{Moretti2000}. 
In our general setup we will prove the following sesqui-symmetry property of the Seeley-DeWitt coefficients (valid not only in Riemannian or Lorentzian, but in arbitrary nondegenerate signature).

\begin{thm}\label{thm-1} For $x$ in  a geodesically convex (in the sense of \cite{Moretti1999}) neighbourhood of point $x'$
 $f_n^{SD}(x,x')=f_n^{SD}(x',x)^\dagger$.
\end{thm}

In the case of the real scalar field (considered and proved in \cite{Moretti2000}) $f_n^{SD}$ are real functions, thus they are just symmetric. 
The proof in \cite{Moretti2000} (see also \cite{Hack2012}) relies on the result about coefficients in the asymptotic expansion of the Riemannian heat kernel \cite{Moretti1999} (let us notice that the heat kernel is symmetric, but symmetry of the coefficients in the expansion is a more subtle issue). One then applies local version of the Wick rotation to obtain the theorem for analytic Lorentian manifolds. Finally, a suitable approximation of a smooth manifold by analytic ones and a special property of differential equation defining the coefficients give the proof. 

The following a bit weaker property (see \cite{friedlander, Garabedian}), that is however sufficient in application in QFT \cite{Wald} was known earlier. In globally hyperbolic spacetimes advanced and retarded propagators are unique and one is adjoint to another. From this property one can deduce symmetry of the Hadamard function $V$ defined inside the lightcone by advanced or retarded propagator. The expansion in $\sigma$ around the light cone is given in $4$ dimensions by the formula
\begin{equation}
V\approx\frac{1}{(4\pi)^{2}} \sum_{n=0}\frac{(-1)^{n+1}}{2^{n+1}n!}f_{n+1}^{SD}\sigma^n.
\end{equation}
From Taylor expansion one can infer symmetry of the coefficients on the lightcones. 

The statement in \cite{Moretti2000} is obtained by analytic approximation and the proof depends on the signature of the metric.  The fact itself is very fundamental, so it is not satisfactory, that there is no elementary proof in the literature.
The goal of the present paper is to provide such a proof, that works in arbitrary signature (not only Riemannian or Lorentzian).

\section{Generalization}

In this section we will describe a bit stronger version of theorem \ref{thm-1}.
The Seeley-DeWitt coefficients $f_n^{SD}(x,x')$ are defined recursively by so-called Hadamard  transport differential equations along affinely parametrized geodesic $\gamma$ connecting points $x'$ and $x$ (we will describe them in details in section \ref{sec:def})
\begin{equation}
\gamma(0)=x',\ \gamma(1)=x.
\end{equation}
They are well-defined as long as there are \textbf{no points  conjugated to $x'$ along the geodesic} i.e.
\begin{equation}\label{cond-conjug}
 \forall_{s\in(0,1]} \left.\det \frac{\partial \gamma(s)}{\partial\dot{\gamma}(0)}\right|_{\gamma(0)}\not=0,
\end{equation}
where $\gamma(s)$ is regarded as a function of initial position $\gamma(0)$ and initial velocity $\dot{\gamma}(0)$ of the geodesic. This condition is true, in particular, if the points are sufficiently close (in precise terms $(x,x')$  belongs to some geodesically convex neighbourhood of the diagonal $\{(x,x)\colon x\in{\mathcal M}\}\subset\mathcal{M}\times \mathcal{M}$).

In general, functions $f_{\gamma\ n}^{SD}$ depend on the choice of $\gamma$ if there are more then one geodesic connecting two points. In fact we expect symmetry also in general case under condition that both $f_{\gamma\ n}^{SD}(x,x')$ and $f_{\gamma^{-1}\ n}^{SD}(x',x)$ are well-defined:
\begin{equation}\label{cond-conjug2}
 \forall_{s\in(0,1]} \left.\det \frac{\partial \gamma(s)}{\partial\dot{\gamma}(0)}\right|_{\gamma(0)}\not=0,\quad
  \forall_{s\in[0,1)} \left.\det \frac{\partial \gamma(s)}{\partial\dot{\gamma}(1)}\right|_{\gamma(1)}\not=0.
\end{equation}
Here we denote inverse geodesic by $\gamma^{-1}$ . However, our theorem is proven under a bit stronger assumption.

We consider manifold $({\mathcal M},g)$ of arbitrary nondegenerate signature and the operator $\KG$ \eqref{eq-KG}:

\begin{thm}\label{thm-2}
Let $\gamma$ be a geodesic connecting $x'$ and $x$ ($\gamma(0)=x'$, $\gamma(1)=x$) such that no two points on the geodesic between $x$ and $x'$ are conjugated:
\begin{equation}\label{cond-conjug3}
 \forall_{s\not=t,\ s,t\in [0,1]} \left.\det \frac{\partial \gamma(s)}{\partial\dot{\gamma}(t)}\right|_{\gamma(t)}\not=0,
\end{equation}
then $f_{\gamma\ n}^{SD}(x,x')=f_{\gamma^{-1}\ n}^{SD}(x',x)^\dagger$.
\end{thm}

In fact for Riemannian signature \cite{Morse} and causal (timelike or null) geodesics in Lorentzian signature \cite{global-lor, Beem} both conditions \eqref{cond-conjug3} and \eqref{cond-conjug2} follow from \eqref{cond-conjug}.  It is probably not true anymore in general (see \cite{Helfer} for a theory of conjugate points in arbitrary signature), so we will keep the stronger assumptions. 
Let us consider isometric embedding of the small open tubular neighbourhood of the geodesic $\gamma$ connecting point $x'$ with $x$
\begin{equation}
\tilde{\mathcal{M}}\hookrightarrow \mathcal{M}.
\end{equation}
For the geodesic $\gamma$ satisfying \eqref{cond-conjug3}, $\tilde{\mathcal M}$ can be chosen small enough such that every two points in $\tilde{\mathcal M}$ are connected by at most one geodesic in $\tilde{\mathcal M}$ that in addition satisfies \eqref{cond-conjug3}. We define a set
\begin{equation}
\tilde{U}=\{(y,y')\in \tilde{\mathcal{M}}\times \tilde{\mathcal{M}}\colon \exists \text{ a geodesic }\gamma'\colon[0,1]\rightarrow  \tilde{\mathcal{M}},\ \gamma'(0)=y', \gamma'(1)=y\}
\end{equation}
This set has the following properties
\begin{enumerate}
\item It is open because endpoints of every geodesic in $\tilde{M}$ can be perturbed by \eqref{cond-conjug3}.
\item Every pair of points of geodesic $\gamma$ belongs to $\tilde{U}$, i.e. $\gamma([0,1])\times\gamma([0,1])\subset \tilde{U}$.
\item Every pair $(y,y')\in \tilde{U}$ is connected by exactly one geodesic $\gamma'$ in $\tilde{M}$ and
\begin{equation}
\forall_{s,t\in[0,1]}\quad (\gamma'(s),\gamma'(t))\in \tilde{U}.
\end{equation}
\end{enumerate}
As the Hadamard transport differential equation for coefficients is defined locally along geodesics we can forget about our initial manifold and work exclusively in $\tilde{U}$. In the case when ${\mathcal M}$  is a geodesically convex neighbourhood of a point we can take $\tilde{\mathcal M}={\mathcal{M}}$ and $\tilde{U}={\mathcal M}\times {\mathcal M}$ thus we are back in the setup of theorem \ref{thm-1}.

\section{Notation conventions} 
\label{sec:notation}

Let us briefly describe our conventions. We are working with functions on $\tilde{U}\subset{\mathcal M}\times {\mathcal M}$ with values in certain bundles. We will consider in fact two cases
\begin{itemize}
\item The bundle is $\pi_1^*({\mathcal V})\otimes \pi_2^*({\mathcal V}^*)$ where $\pi_1$ and $\pi_2$ are projections onto first and second copy of the manifold and ${\mathcal V}^*$ is a dual bundle. At given point $(x,x')\in {\mathcal M}\times {\mathcal M}$ such a function $f$ takes values in linear operators from fiber over $x'$ to fiber over $x$ thus we write
\begin{equation}
f(x,x')\colon {\mathcal V}_{x'}\rightarrow {\mathcal V}_x
\end{equation}
\item The bundle is $\pi_2^*({\mathcal V})\otimes \pi_2^*({\mathcal V}^*)=\pi_2^*(\End {\mathcal V})$  and then we write
\begin{equation}
g(x,x')\colon {\mathcal V}_{x'}\rightarrow {\mathcal V}_{x'}
\end{equation}
\end{itemize}
In addition to functions on $\tilde{U}\subset{\mathcal M}\times {\mathcal M}$ we will also use smooth functions on $\tilde{U}\times \R\subset{\mathcal M}\times {\mathcal M}\times \R$ denoted by
\begin{equation}
F(x,x',\lambda)\colon \mathcal{V}_{x'}\rightarrow \mathcal{V}_{x},\quad G(x,x',\lambda)\colon \mathcal{V}_{x'}\rightarrow \mathcal{V}_{x'},
\end{equation}
We are mainly interested in these functions at $\lambda=0$ but up to any order in the derivative in $\lambda$. We will write $O(\lambda^\infty)$ for a function such that its any derivative in $\lambda$ vanishes for $\lambda=0$
\begin{equation}
h=O(\lambda^\infty) \Longleftrightarrow \forall_n \partial_\lambda^n h|_{\lambda=0}=0
\end{equation}
For a smooth kernel $F(x,x',\lambda)\colon \mathcal{V}_{x'}\rightarrow \mathcal{V}_{x}$  we introduce $*$ by
\begin{equation}
F^*(x,x',\lambda)=F(x',x,-\lambda)^\dagger.
\end{equation}
We will use this notation also for kernels independent of $\lambda$ that is
\begin{equation}
f^*(x,x')=f(x',x)^\dagger.
\end{equation}
We will consider differential operators on functions on ${\mathcal M}\times {\mathcal M}$ pull backed by $\pi_1$ (action on the first copy of the manifold) and by $\pi_2$ (action on the second copy). Let us consider a differential operators $D_1\colon \Gamma({\mathcal M}, {\mathcal V})\rightarrow \Gamma({\mathcal M}, {\mathcal V})$ and $D_2\colon \Gamma({\mathcal M}, {\mathcal V}^*)\rightarrow \Gamma({\mathcal M}, {\mathcal V}^*)$ We will use notation
\begin{equation}
D_1f:= (D_1\otimes{\mathbb I})f,\quad D_1F:=(D_1\otimes{\mathbb I}\otimes {\mathbb I})F,\quad D_2'f:=({\mathbb I}\otimes D_2)f,\quad D_2'F:=({\mathbb I}\otimes D_2\otimes {\mathbb I})F.
\end{equation}
Let us denote $D^{dual}\colon \Gamma({\mathcal M}, {\mathcal V}^*)\rightarrow \Gamma({\mathcal M}, {\mathcal V}^*)$ the dual differential operator to $D$ defined by equality on compactly supported smooth functions $\psi\in \Gamma({\mathcal M}, {\mathcal V}^*)$, $\phi\in \Gamma({\mathcal M}, {\mathcal V})$ by
$\int_{\mathcal M}(D^{dual}\psi, \phi)\sqrt{|\det g|}d^dx=\int_{\mathcal M}(\psi, D\phi)\sqrt{|\det g|}d^dx$. We have the following identity
\begin{equation}
D_{dual}'F=(D^\dagger F^*)^*.
\end{equation}
In order to simplify the notation we will adopt short version of covariant derivative operation 
\begin{equation}
f_{;\mu}=\nabla_\mu f,\quad f_{;'\mu}=\nabla_\mu' f,\quad F_{;\mu}=\nabla_\mu F,\quad F_{;'\mu}=\nabla_\mu' f,\quad 
\end{equation}
Similar notations apply to $G$ and $g$ functions. Let us stress that our functions $G$, $g$ are scalars from the point of view of first copy of $\mathcal M$ thus we can define
$G_{,\mu}$, $g_{,\mu}$. 
We will also use derivative with respect to third component in cartesian product
\begin{equation}
\partial_\lambda F:= ({\mathbb I}\otimes {\mathbb I}\otimes \partial)F
\end{equation}
Actions of differential operators on first copy of the manifold commute with actions on the second copy as well as with $\partial_\lambda$. 

We denote $g_{\mu\nu}$, $g^{\mu\nu}$, $\det g$ the metric, inverse metric and determinant of the metric respectively.
In the paper Synge notation of taking coinciding points $\lfloor F\rfloor(x')=F(x',x')$ and  $\lfloor F \rfloor (x',\lambda)=F(x,',x',\lambda)$ will be used. Often we will omit writing $x,x',\lambda$ if arguments are obvious from the context. We will also follow Einstein summation convention.

The value of the functions under consideration for a fixed argument is a linear operator between suitable fibers. Products in the formulas will be regarded as a proper composition of linear operators multiplied by scalars. For example
$f_{;\mu}(x,x')c(x')\sigma(x,x')^n$ (where $c(x')\colon {\mathcal V}_{x'}\rightarrow {\mathcal V}_{x'}$) is a composition of covariant derivative of $f$, that is a linear operator ${\mathcal V}_{x'}\rightarrow {\mathcal V}_{x}$, with $c$ multiplied by $\sigma(x,x')^n$.

\section{Local definitions of the Seeley-DeWitt coefficients}
\label{sec:def}

On a manifold $\mathcal M$ we define a geodesic map
\begin{equation}\label{map-can}
 (\gamma(0),\dot{\gamma}(0))\rightarrow (\gamma(1),\dot{\gamma}(1))
\end{equation}
associating final position $\gamma(1)$ and final velocity $\dot{\gamma}(1)$ of the affinely parametrized geodesic $\gamma$  to its initial position and velocity.
If the manifold is not geodesically complete, which is often the case, then the map is defined only for an open subset of velocities and positions. However, this is enough for what we need in our considerations.
The map \eqref{map-can}  is in fact a canonical transformation related to the Hamilton variation principle for the action
\begin{equation}
 S=\int_0^1 ds\ p_{\mu_s} \dot{\gamma}^{\mu_s}-\frac{1}{2}g^{\mu_s\nu_s}(\gamma)p_{\mu_s} p_{\nu_s}
\end{equation}
where one of the equations of motion is $p_{\mu_s}=\dot{\gamma}_{\mu_s}$. 
Let us assume that locally we can determine $\dot{\gamma}(0)$  from 
\begin{equation}
 (\gamma(0),\gamma(1))
\end{equation}
Consequently, locally we can use $x'=\gamma(0)$ and $x=\gamma(1)$ as variable parametrizing geodesics. One says that
\textbf{${\mathbf x'}$ and ${\mathbf x}$ are not conjugated along geodesic  ${\mathbf \gamma}$}. This condition  can be written as
\begin{equation}
\left. \det \frac{\partial \gamma(1)}{\partial\dot{\gamma}(0)}\right|_{\gamma(0)}\not=0.
\end{equation}
We can introduce Hamilton generating function (that maps $(x,x')$ into the value of the action on the solution of equations of motion with $x=\gamma(0)$ and $x'=\gamma(1)$)
\begin{equation}
 \sigma(x,x')=\frac{1}{2}\int_0^1 ds\ \dot{\gamma}_{\mu_s} \dot{\gamma}^{\mu_s},\quad \gamma(0)=x',\ \gamma(1)=x.
\end{equation}
In fact $\sigma$ is half of the (signed) square geodesic distance. Moreover, the following holds
\begin{equation}
 \sigma_{;\mu}=p_{\mu}(1)=\dot{\gamma}(1)_\mu,\quad \sigma_{;'\mu'}=-p_{\mu'}(0)=-\dot{\gamma}(0)_{\mu'}.
\end{equation}
Thus we can write
\begin{equation}
\left. \det\frac{\partial \gamma(1)}{\partial\dot{\gamma}(0)}\right|_{\gamma(0)}=(\det \left(-\sigma_{;\mu;'\mu'}\right))^{-1}.
\end{equation}
The determinant $\det \left(-\sigma_{;\mu;'\mu'}\right)$ is a density in both $x$ and $x'$. Comparing to the canonical measures given by the metric tensor we can introduce a function $\Delta$ known as Van Vleck-Morette determinant\footnote{Sometimes associated also with van Hove and Pauli \cite{Choquard}.}
\begin{equation}
 \Delta(x,x')=\frac{\det \left(-\sigma_{;\mu;'\mu'}\right)}{\sqrt{|\det g|}\sqrt{|\det g'|}}
\end{equation}
Sign of $\Delta$ is $(-1)^{\frac{d+\sign g}{2}}$ where $\sign g$ is the signature of the metric. Van Vleck-Morette determinant is a symmetric scalar function. We will adopt also notation
\begin{equation}
\Delta^{\frac{1}{2}}(x,x')=\sqrt{| \Delta(x,x')|},
\end{equation}
although the actual square root is defined up to a sign and depends on the signature. 

The modified parallel transport $H(x,x')\colon \mathcal{V}_{x'}\rightarrow \mathcal{V}_x$ along the geodesic $\gamma$ connecting point $x'$ with $x$ is defined by an equation along this geodesic 
\begin{equation}\label{H-def}
\dot{\gamma}^\mu \left(H_{;\mu}+A_\mu H\right)=0,\quad H(x',x')={\mathbb I}.
\end{equation}
The parallel transport satisfies
\begin{equation}\label{H-sym}
H(x,x')=H(x',x)^\dagger,\quad H(x,x')H(x',x)={\mathbb I},
\end{equation}
because $A$ is anti-hermitian and the geodesic $\gamma^{-1}$ connecting point $x$ with $x'$ is inverse of the geodesic $\gamma$. In particular, $H$ is an invertible matrix for all $x$ and $x'$.

We can now rewrite the Seeley-DeWitt coefficients introducing coefficients $g_n^{SD}(x,x')\colon \mathcal{V}_{x'}\rightarrow \mathcal{V}_{x'}$ (scalar from the point of view of the first copy of the manifold) by
\begin{equation}
f_n^{SD}(x,x')= \Delta^{\frac{1}{2}}(x,x')H(x,x')g_n^{SD}(x,x'),
\end{equation}
where we use our convention about the product ($\Delta^{\frac{1}{2}}(x,x')$ is a scalar function).

Let us consider the following \textbf{Hadamard recursion transport equations}
\begin{equation}
\sigma(x,x')_{,\mu}g_n^{SD}(x,x')^{,\mu}+ng_n^{SD}(x,x')- J(g_{n-1}^{SD})(x,x')=0,
\end{equation}
with operator $J$ defined by
\begin{equation}
J(g)(x.x')=\Delta^{\frac{1}{2}}(x,x')^{-1}H(x,x')^{-1}\KG(\Delta^{\frac{1}{2}}(x,x')H(x,x')g(x,x')),
\end{equation}
where for $n=0$ we assume that $g_{n-1}^{SD}=0$. This is a recursion definition by a transport equation. Due to its special behaviour at $x=x'$ for $n>0$ this equation has a unique smooth solution \cite{Hormander-III, Baer}\footnote{Different factors are due to difference between Seeley-DeWitt and Hadamard coefficients.}, defined inductively by
\begin{equation}\label{transport-1}
g_n^{SD}(x,x')=\int_0^1 ds\ s^{n-1} J(g_{n-1}^{SD})(\gamma(s),x'),
\end{equation}
where $\gamma$ is the geodesic starting from point $x'$ and arriving to point $x$ at affine time $1$.

For $n=0$ the equation implies only that function $g_0^{SD}(x,x')$ is independent of $x$ and we assume it to be
\begin{equation}\label{transport-2}
g_0^{SD}(x,x')={\mathbb I}.
\end{equation}

The equations and thus also coefficients $f_n^{SD}$ are well-defined as long as 
condition \eqref{cond-conjug} is satisfied. This condition means that $x'$ is not conjugated to any $\gamma(s)$, $s\in (0,1]$ along geodesic $\gamma$.
Functions $\Delta$, $\sigma$ and $f_n^{SD}$ depend on the choice of the geodesic $\gamma$, thus we should rather write
\begin{equation}
\Delta_\gamma,\quad \sigma_\gamma,\quad f_{\gamma\ n}^{SD}.
\end{equation}
However, we will now assume that condition \eqref{cond-conjug3} is satisfied and we consider coefficients as defined in $\tilde{U}$. They are then uniquely determined and we skip index $\gamma$. 

\section{Reformulation}

Let us consider a smooth function $F^{SD}(x,x',\lambda)\colon {\mathcal V}_{x'}\rightarrow {\mathcal V}_{x}$ such that
\begin{equation}
\forall_{n\geq 0}\quad  \frac{1}{i^nn!}\partial_\lambda^n F^{SD}(x,x',0)=f_n^{SD}(x,x').
\end{equation}
In fact such a function always exists by Borel theorem (see \cite{Treves} theorem 37.2 or Appendix \ref{Appendix}), although it is not unique. It is defined up to $O(\lambda^\infty)$ thus we will write
\begin{equation}
\label{FG-notation}
 {F}^{SD}(x,x',\lambda)=\sum_{n=0}^\infty i^n{f}_n^{SD}(x,x')\lambda^{n} +O(\lambda^\infty),
\end{equation}
even if the sum is not convergent.

In order to give a broader contex of our method let us consider for a moment the following kernel (see \cite{Birrell})
\begin{equation}\label{exp}
K(x,x',\lambda)=\frac{1}{(4\pi)^{\frac{d}{2}}}e^{i\frac{\pi}{4}\sgn\lambda\sign g}F^{SD}(x,x',\lambda)|\lambda|^{-\frac{d}{2}}e^{\frac{i}{2\lambda}\sigma(x,x')}.
\end{equation}
It has interesting properties. Namely, for any compactly supported (support in $\tilde{U}\cap\{x\}\times\tilde{\mathcal{M}}$) smooth function $\phi\in \Gamma(\tilde{\mathcal M},{\mathcal V})$
\begin{enumerate}
\item $\displaystyle\lim_{\lambda\rightarrow 0} \displaystyle\int K(x,x',\lambda)\phi(x') \sqrt{|\det g'|}d^dx'=\phi(x)$
\item $\displaystyle\int (\KG+i\partial_\lambda)K(x,x',\lambda)\phi(x') \sqrt{|\det g'|}d^dx'=O(\lambda^\infty)$
\end{enumerate}
Both equations follow from application of the stationary phase method. For the second one needs to notice (see lemma \ref{lm-1} and \ref{lm-2} below) that $(\KG+i\partial_\lambda)K(x,x',\lambda)=O(\lambda^\infty)e^{\frac{i}{2\lambda}\sigma(x,x')}$.

Let us also assume that $\KG$ is essentially self-adjoint and the scalar product in the bundle is positively defined. In this case $U_\lambda=e^{i\lambda \KG}$ is unitary thus
\begin{equation}
U_\lambda^\dagger=U_{-\lambda}
\end{equation}
If $K$ is a pointwise approximation of $U_\lambda$ up to any order in $\lambda$ then  multiplying both $K$ and $K^*$ by $e^{-\frac{i}{2\lambda}\sigma}={e^{-\frac{i}{2\lambda}\sigma}}^*$ (of absolute value $1$) and comparing Taylor expansions in $\lambda$ would show sesqui-symmetry of $f_n^{SD}$ coefficients.

However, it is not the case in general (see for example \cite{Duistermaat}) and we need to develop a different approach. Although \eqref{exp} is not a proper approximation, it is determined almost uniquely (up to $O(\lambda^\infty)$ term in $F^{SD}$) by certain properties. It is thus enough to show that $K^*$ satisfies these properties too.

\section{Proof of the sesqui-symmetry of the Seeley-DeWitt coefficients}

In this section we will present the proof of theorem \ref{thm-2}. For the sake of clarity we will postpone proofs of technical lemmata to the next sections.

Let us intoduce the operator $P$ on a functions on $\tilde{U}\times\R\subset{\mathcal M}\times {\mathcal M}\times \R$ defined in the following way:
\begin{enumerate}
\item For function $F(x,x',\lambda)\colon {\mathcal V}_{x'}\rightarrow {\mathcal V}_{x}$ we define $G(x,x',\lambda)\colon {\mathcal V}_{x'}\rightarrow {\mathcal V}_{x'}$ by equality
\begin{equation}\label{G-def}
F(x,x',\lambda)=\Delta^{\frac{1}{2}}(x,x')H(x,x')G(x,x',\lambda).
\end{equation}
\item Define $P(F)(x,x',\lambda)\colon {\mathcal V}_{x'}\rightarrow {\mathcal V}_{x}$ by
\begin{equation}\label{P-def}
P(F)=\Delta^{\frac{1}{2}}H\left(
\sigma_{,\mu} G^{,\mu}+\lambda\partial_\lambda G-i\lambda J(G)\right).
\end{equation}
\end{enumerate}

We will also consider $P'$ given by
\begin{equation}\label{Pp-def}
P'(F)=P(F^*)^*.
\end{equation}

\begin{lm}\label{lm-1}
Let us suppose that $P(F)=O(\lambda^\infty)$ then
\begin{equation}
F(x,x',\lambda)=F^{SD}(x,x',\lambda)c(x')+O(\lambda^\infty)
\end{equation}
where $c(x')=\lfloor F\rfloor (x',0)\colon {\mathcal V}_{x'}\rightarrow {\mathcal V}_{x'}$. In particular $P(F^{SD})=O(\lambda^\infty)$.
\end{lm}

Let us now describe a way of attacking the problem. We will show that
\begin{equation}
P({F^{SD}}^*)=O(\lambda^\infty).
\end{equation}
This together with $ \lfloor {F^{SD}}^*\rfloor (x,0)=\lfloor {F^{SD}}\rfloor (x,0)^\dagger={\mathbb I}$ gives due to the lemma \ref{lm-1}
\begin{equation}
{F^{SD}}^*={F^{SD}}+O(\lambda^\infty).
\end{equation}
Comparing order by order Taylor expansion we get
\begin{equation}
f_n^{SD}(x',x)^\dagger=f_n^{SD}(x,x').
\end{equation}
This is exactly the sesqui-symmetry condition.

In order to accomplish this goal we need special property of $P$ and $P'$. Not complicated but a bit tedious computation leads to the next lemma.
\begin{lm}\label{lm-2}
For $\lambda\not=0$ and $F(x,x',\lambda)\colon {\mathcal V}_{x'}\rightarrow {\mathcal V}_{x}$ smooth function 
\begin{align}\label{eq-important1}
P(F)&=-i\lambda |\lambda|^{\frac{d}{2}}e^{-\frac{i}{2\lambda}\sigma}\left[\KG +i\partial_\lambda \right]\left( F|\lambda|^{-\frac{d}{2}}e^{\frac{i}{2\lambda}\sigma}\right),\\
\label{eq-important2}P'(F)&=-i\lambda |\lambda|^{\frac{d}{2}}e^{-\frac{i}{2\lambda}\sigma}\left[\KG_{dual}' +i\partial_\lambda \right]\left( F|\lambda|^{-\frac{d}{2}}e^{\frac{i}{2\lambda}\sigma}\right),
\end{align}
\end{lm}

This characterization  is related to the approximation property discussed earlier. It allow us to show

\begin{lm}
For $\lambda\not=0$ and $F(x,x',\lambda)\colon {\mathcal V}_{x'}\rightarrow {\mathcal V}_{x}$ smooth function the following identity holds
\begin{equation}
P\left(\frac{i}{\lambda}P'(F)\right)=P'\left(\frac{i}{\lambda}P(F)\right).
\end{equation}
\end{lm}

\begin{proof}
From \eqref{eq-important1},\eqref{eq-important2} the left and right hand side are equal to
\begin{align}
P\left(\frac{i}{\lambda}P'(F)\right)=&-i\lambda |\lambda|^{\frac{d}{2}}e^{-\frac{i}{2\lambda}\sigma}\left[\KG+i\partial_\lambda\right]\left[ \KG_{dual}'+i\partial_\lambda \right]\left( F|\lambda|^{-\frac{d}{2}}e^{\frac{i}{2\lambda}\sigma}\right),\\
P'\left(\frac{i}{\lambda}P(F)\right)=&-i\lambda |\lambda|^{\frac{d}{2}}e^{-\frac{i}{2\lambda}\sigma}\left[\KG_{dual}'+i\partial_\lambda\right]\left[ \KG+i\partial_\lambda \right]\left( F|\lambda|^{-\frac{d}{2}}e^{\frac{i}{2\lambda}\sigma}\right).
\end{align}
Equality follows because operators acting on the first copy of the manifold commute with those acting on the second copy, as well as any such operator commute with $\partial_\lambda$.
\end{proof}

Suppose that the smooth function $F(x,x',\lambda)$ has the property $F(x,x',0)=0$ then
\begin{equation}
\frac{i}{\lambda} F(x,x',\lambda)
\end{equation}
is also a smooth function (i.e. it extends smoothly to $\lambda=0$). In particular $\frac{i}{\lambda}O(\lambda^\infty)=O(\lambda^\infty)$.

%By elementary computation we can show the next lemma.

\begin{lm}\label{lm-4}
The following holds
\begin{equation}
P'(F^{SD})(x,x',0)=0,\quad \lfloor \partial_\lambda P'(F^{SD})\rfloor (x',0)=0
\end{equation}
\end{lm}

We see that both
\begin{equation}
\frac{i}{\lambda} P'(F^{SD}),\quad \frac{i}{\lambda} P(F^{SD})
\end{equation}
are smooth functions as $P(F^{SD})=O(\lambda^\infty)$. Hence,
\begin{equation}
P\left(\frac{i}{\lambda} P'(F^{SD})\right)=P'\left(\frac{i}{\lambda} P(F^{SD})\right)=
P'(O(\lambda^\infty))=O(\lambda^\infty)
\end{equation}
and we obtain by lemma \ref{lm-1}
\begin{equation}
\frac{i}{\lambda} P'(F^{SD})(x,x',\lambda)=F^{SD}(x,x',\lambda)c(x')+O(\lambda^\infty)
\end{equation}
where $c(x')=\left\lfloor \frac{i}{\lambda} P'(F^{SD})\right\rfloor (x',0)=i\lfloor \partial_\lambda P'(F^{SD})\rfloor (x',0)=0$ and thus,
\begin{equation}
\frac{i}{\lambda} P'(F^{SD})(x,x',\lambda)=O(\lambda^\infty)\Rightarrow P({F^{SD}}^*)=O(\lambda^\infty).
\end{equation}
This finishes the proof of the theorem \ref{thm-2}.

\section{Characterization of the coefficients (proof of lemma \ref{lm-1})}

Let us recall the smooth function $G(x,x',\lambda)\colon {\mathcal V}_{x'}\rightarrow {\mathcal V}_{x'}$ defined by equality $F=\Delta^{\frac{1}{2}}H G$ \eqref{G-def}. We introduce its Taylor expansion in $\lambda$
\begin{equation}
G(x,x',\lambda)=\sum_{n=0}^\infty i^n\lambda^n g_n(x,x')+O(\lambda^\infty).
\end{equation}
Condition $P(F)=O(\lambda^\infty)$ expanded order by order in $\lambda$ gives Hadamard recursion relations for the coefficients $g_n$ (we use convention that $g_{-1}=0$)
\begin{equation}
\sigma(x,x')_{,\mu}g_n(x,x')^{,\mu}+ng_n(x,x')- J(g_{n-1})(x,x')=0.
\end{equation}
In particular $g_0$ is independent of $x$ thus $g_0(x,x')=\lfloor g_0\rfloor (x')$. From
$\lfloor \Delta^{\frac{1}{2}}\rfloor=1$ (see \cite{Fulling}) and $\lfloor H\rfloor={\mathbb I}$   follows that $\lfloor \Delta^{\frac{1}{2}}H\rfloor={\mathbb I}$. Thus
\begin{equation}
\lfloor g_0\rfloor (x')=\lfloor F\rfloor (x',0).
\end{equation}
Let us now write 
\begin{equation}
g_0(x,x')=g_0^{SD}(x,x')\lfloor g_0\rfloor (x').
\end{equation}
Inductively assuming that $g_{n-1}(x,x')=g_{n-1}^{SD}(x,x')\lfloor g_0\rfloor (x')$ we show by \eqref{transport-1} (uniqueness of the solution holds in $\tilde{U}$)
\begin{equation}
g_n(x,x')=\int_0^1 ds\ s^{n-1} J(g_{n-1})(\gamma(s),x')=
\int_0^1 ds\ s^{n-1} J(g_{n-1}^{SD})(\gamma(s),x')\lfloor g_0\rfloor (x')= g_n^{SD}(x,x')\lfloor g_0\rfloor (x'),
\end{equation}
where we used the fact that $J$ commutes with multiplication by $\lfloor g_0\rfloor (x')$.
Concluding
\begin{equation}
f_n(x.x')=f_n^{SD}(x,x')c(x'),
\end{equation}
where $c(x')=\lfloor g_0\rfloor (x')=\lfloor F\rfloor (x',0)$. Combining Taylor expansion, it follows
\begin{equation}
F(x,x',\lambda)=\sum_{n=0}^\infty i^n\lambda^n f_n(x,x')+O(\lambda^\infty)=
\sum_{n=0}^\infty i^n\lambda^n f_n^{SD}(x,x')c(x')+O(\lambda^\infty)=
F^{SD}(x,x',\lambda)c(x')+O(\lambda^\infty).
\end{equation}

\section{Formula for $P$ and $P'$ (proof of lemma \ref{lm-2})}

The proof of this lemma is a standard computation. The following identities are useful (see \cite{Fulling, Avramidi1986})
\begin{equation}\label{eq-P}
\sigma^{,\mu}(\Delta^{\frac{1}{2}})_{,\mu}=\left(\frac{d}{2}-\frac{1}{2}\sigma^{;\mu}_{;\mu}\right)\Delta^{\frac{1}{2}},\quad \sigma_{,\mu}\sigma^{,\mu}=2\sigma.
\end{equation}
Let us notice that $(\nabla_\mu+A_\mu)F |\lambda|^{-\frac{d}{2}}e^{\frac{i}{2\lambda}\sigma} =|\lambda|^{-\frac{d}{2}}
e^{\frac{i}{2\lambda}\sigma}\left(\nabla_\mu+\frac{i}{2\lambda}\sigma_{,\mu}+A_\mu\right)F$ thus
\begin{equation}\label{eq-P-1}
\KG \left(F|\lambda|^{-\frac{d}{2}}e^{\frac{i}{2\lambda}\sigma}\right)=|\lambda|^{-\frac{d}{2}}e^{\frac{i}{2\lambda}\sigma}\left(
\KG+\frac{i}{\lambda}\left(\sigma_{,\mu} \nabla^\mu+\sigma^{,\mu} A_\mu+\frac{1}{2}\sigma_{;\mu}^{;\mu}\right)-\frac{1}{4\lambda^2}\sigma_{,\mu}\sigma^{,\mu}\right) F,
\end{equation}
Another important identity is
\begin{equation}\label{eq-P-2}
i\partial_\lambda \left(F|\lambda|^{-\frac{d}{2}}e^{\frac{i}{2\lambda}\sigma}\right)=|\lambda|^{-\frac{d}{2}}e^{\frac{i}{2\lambda}\sigma}\left(\frac{1}{2\lambda^2}\sigma -\frac{i}{\lambda}\frac{d}{2}+i\partial_\lambda \right)F.
\end{equation}
Adding \eqref{eq-P-1} to \eqref{eq-P-2} and using \eqref{eq-P} we obtain
\begin{equation}\label{eq-P-3}
(\KG +i\partial_\lambda)\left(F|\lambda|^{-\frac{d}{2}}e^{\frac{i}{2\lambda}\sigma}\right)=|\lambda|^{-\frac{d}{2}}e^{\frac{i}{2\lambda}\sigma}\left(
\KG+i\partial_\lambda+\frac{i}{\lambda}\left(\sigma_{,\mu} \nabla^\mu+\sigma^{,\mu} A_\mu+\frac{1}{2}\sigma_{;\mu}^{;\mu}-\frac{d}{2}\right)\right) F.
\end{equation}
In order to derive the result we need the following identity 
\begin{equation}
\sigma_{,\mu} \nabla^\mu\left(\Delta^{\frac{1}{2}} H G\right)=
\sigma^{,\mu} \left(\Delta^{\frac{1}{2}} H G_{,\mu}+\Delta^{\frac{1}{2}} H_{;\mu} G+\Delta^{\frac{1}{2}}_{,\mu} H G\right)=
\sigma^{,\mu} \left(\Delta^{\frac{1}{2}} H G_{,\mu}-\Delta^{\frac{1}{2}} A_{\mu}H G\right)+\left(\frac{d}{2}-\frac{1}{2}\sigma_{;\mu}^{;\mu}\right)\Delta^{\frac{1}{2}} H G
\end{equation}
From the last equation we get
\begin{equation}\label{eq-P-4}
\left(\sigma_{,\mu} \nabla^\mu+\sigma^{,\mu} A_\mu+\frac{1}{2}\sigma_{;\mu}^{;\mu}-\frac{d}{2}\right) F=
\Delta^{\frac{1}{2}} H \sigma_{,\mu} G^{,\mu}.
\end{equation}
Applying \eqref{eq-P-4} to \eqref{eq-P-3}
we obtain desired identity for $P$. The case of $P'$ is analogous.

\section{Properties of $P'(F^{SD})$ (proof of lemma \ref{lm-4})}

Let us notice that 
\begin{equation}
P'(F^{SD})^*(x,x',0)=\Delta^{\frac{1}{2}}H\sigma_{,\mu} G_{,\mu}(x,x',0),
\end{equation}
where $G$ is defined by ${F^{SD}}^*=\Delta^{\frac{1}{2}}H G$.
However ${F^{SD}}^*(x,x',0)=(\Delta^{\frac{1}{2}}H)^*=\Delta^{\frac{1}{2}}H=F^{SD}(x,x',0)$ by \eqref{H-sym}. Thus $G(x,x',0)=G^{SD}(x,x',0)={\mathbb I}$ and $P'(F^{SD})(x,x',0)=0$. This proves the first part of the lemma.

From $\lfloor \Delta^{\frac{1}{2}}H\rfloor={\mathbb I}$ and $\lfloor \sigma_{,\mu}\rfloor=\lfloor \sigma_{,'\mu}\rfloor=0$ in \eqref{P-def} and \eqref{Pp-def} it follows that
\begin{equation}
\lfloor P'(F^{SD})\rfloor=\lambda \partial_\lambda \lfloor G^{SD}\rfloor -i\lambda 
\lfloor (\KG {F^{SD}}^*)^*\rfloor=\lambda \partial_\lambda \lfloor G^{SD}\rfloor -i\lambda 
\lfloor \KG {F^{SD}}^*\rfloor^\dagger.
\end{equation}
On the other hand
\begin{equation}
O(\lambda^\infty)=\lfloor P(F^{SD})\rfloor=\lambda \partial_\lambda \lfloor G^{SD}\rfloor -i\lambda 
\lfloor \KG F^{SD}\rfloor,
\end{equation}
so 
\begin{equation}
\lfloor P'(F^{SD})\rfloor=i\lambda\left(\lfloor \KG F^{SD}\rfloor-\lfloor \KG {F^{SD}}^*\rfloor^\dagger\right)+O(\lambda^\infty)\Rightarrow \partial_\lambda\lfloor P'(F^{SD})\rfloor(x,0)=i\lfloor \KG \Delta^{\frac{1}{2}}H\rfloor(x)-i\lfloor \KG \Delta^{\frac{1}{2}}H\rfloor(x)^\dagger.
\end{equation}
where we used $f_0^{SD}={f_0^{SD}}^*=\Delta^{\frac{1}{2}}H$.
 In order to prove that $\partial_\lambda\lfloor P'(F^{SD})\rfloor(x,0)=0$ we need the following lemma.

\begin{lm}\label{lm-5}
$\lfloor\KG(\Delta^{\frac{1}{2}}H)\rfloor=\lfloor\KG (\Delta^{\frac{1}{2}}H)\rfloor^\dagger$
\end{lm}

\begin{proof}
From \eqref{H-def}, $\lfloor\Delta^{\frac{1}{2}}H\rfloor={\mathbb I}$  and $\lfloor\Delta^{\frac{1}{2}}_{,\mu}\rfloor=0$ (see \cite{Fulling}) we derive $\lfloor(\Delta^{\frac{1}{2}}H)_{;\mu}\rfloor=-A^\mu$.
Thus taking derivative
\begin{equation}
-A_{\mu;\nu}=\lfloor(\Delta^{\frac{1}{2}}H)_{;\mu}\rfloor_{;\nu}=
\lfloor(\Delta^{\frac{1}{2}}H)_{;\mu;\nu}+(\Delta^{\frac{1}{2}}H)_{;\mu;'\nu}\rfloor,
\end{equation}
contracting with $g^{\mu\nu}$ we obtain
$\lfloor(\Delta^{\frac{1}{2}}H)_{;\mu}^{;\mu}\rfloor=-\lfloor(\Delta^{\frac{1}{2}}H)_{;\mu}^{;'\mu}\rfloor-A^{\mu}_{;\mu}$
and
\begin{equation}
\lfloor\KG(\Delta^{\frac{1}{2}}H)\rfloor=
\lfloor (\Delta^{\frac{1}{2}}H)_{;\mu}^{;\mu}\rfloor + 2A^{\mu}\lfloor(\Delta^{\frac{1}{2}}H) _{;\mu}\rfloor+\left(A^\mu_{;\mu}+A_\mu A^\mu+B\right)\lfloor(\Delta^{\frac{1}{2}}H)\rfloor=
-g^{\mu\nu}\lfloor(\Delta^{\frac{1}{2}}H)_{;\mu;'\nu}\rfloor-A_\mu A^\mu+B.
\end{equation}
Due to $(\Delta^{\frac{1}{2}}H)^*=\Delta^{\frac{1}{2}}H$ 
\begin{equation}
\lfloor(\Delta^{\frac{1}{2}}H)_{;\mu;'\nu}\rfloor=
\lfloor(\Delta^{\frac{1}{2}}H)_{;'\mu;\nu}\rfloor^\dagger,
\end{equation}
so $g^{\mu\nu}\lfloor(\Delta^{\frac{1}{2}}H)_{;\mu;'\nu}\rfloor$ is hermitian as well as $A_\mu A^\mu$ and $B$.
\end{proof}

Let us notice that in case of standard Klein-Gordon equation the equality of lemma \ref{lm-5} is trivially satisfied because $\KG (\Delta^{\frac{1}{2}})$ is real and $H=1$.

\section{Summary}

In this note we proved in an elementary way an important property: the Seeley-DeWitt coefficients on smooth manifolds of arbitrary signature are sesqui-symmetric. Let us notice, that in this way we also proved that Hadamard, as well as DeWitt-Schwinger coefficients are symmetric as they are linear combinations of Seeley-DeWitt coefficients $f_n^{SD}$ (see \cite{Hack2012}).

\begin{acknowledgments}
Author would like to thank Jan Derezi{\'n}ski and Daniel Siemssen for useful discussions. This work was partially supported by BST.
\end{acknowledgments}

\appendix
\section{Taylor expansion}
\label{Appendix}

For convencience of the reader we will provide a proof of the following fact, a version of Borel theorem:

\begin{lm}
Let ${\mathcal N}$ be a Hausdorff, paracompact, smooth manifold and ${\mathcal W}$ a vector bundle on ${\mathcal N}$.
Suppose that we have smooth sections
\begin{equation}
h_n\in\Gamma({\mathcal N},{\mathcal W}),\quad n\geq 0,
\end{equation}
then there exists a smooth section of the pull-backed bundle to ${\mathcal N}\times \R$
\begin{equation}
H\in\Gamma({\mathcal N}\times\R,{\mathcal W}),
\end{equation}
such that
\begin{equation}
\frac{1}{i^nn!} \left.\partial_\lambda^n H(x,\lambda)\right|_{\lambda=0}=h_n(x).
\end{equation}
\end{lm}

\begin{proof}
We first construct such a function on an open set $O\subset \R^d$ with compact closure and such that the bundle is trivial. In this case it is enough to restrict to the case of scalar functions.

Let us introduce a smooth function $\chi\colon \R\rightarrow \R$ with the property that
$\chi(x)=1$ for $x\in\left[-\frac{1}{2},\frac{1}{2}\right]$ and $\chi(x)=0$ for $|x|\geq 1$.
Let us also denote 
\begin{equation}
C_k=\sup_{\lambda\in[-1,1], 0\leq m\leq k} \left|\partial_\lambda^m\chi(\lambda)\right|,\quad
L_n=\sup_{x\in O, 0\leq|\alpha|\leq n} \left| \partial_x^\alpha h_n(x)\right|,
\end{equation}
where $\alpha$ denote multi-index with number of indices $|\alpha|$.

Define $\lambda_n=\min\left\{\frac{1}{n+1}, \frac{1}{L_n}\right\}$ and
\begin{equation}
H(x,\lambda)=\sum_{n=0}^\infty i^n\chi\left(\frac{\lambda}{\lambda_n}\right)\lambda^n h_n(x).
\end{equation}
We will show that every derivative of the series is uniformly convergent thus the result is a well-defined smooth function. 

Let us notice that 
\begin{equation}
\left|\partial_\lambda^k \left(i^n\chi\left(\frac{\lambda}{\lambda_n}\right)\lambda^n\right)\right|\leq C_k (n+1)^k\lambda_n^{n-k},
\end{equation}
and for $n\geq |\alpha|$
\begin{equation}
\left|\partial_x^\alpha\partial_\lambda^k \left(i^n\chi\left(\frac{\lambda}{\lambda_n}\right)\lambda^n\right)h_n(x)\right|\leq C_k (n+1)^k\lambda_n^{n-k}L_n=
C_k (n+1)^k\lambda_n^{n-k-1}\ \lambda_nL_n
\leq C_k (n+1)^{2k+1-n},
\end{equation}
hence the series of derivatives is uniformly convergent and $H$ is well-defined and smooth.
Moreover $\partial_\lambda^k \left.\chi\left(\frac{\lambda}{\lambda_n}\right)\right|_{\lambda=0}=0$ and the Taylor expansion is as desired.
 
 The global result is obtained by patching local functions multiplied by locally finite partition of unity $\{O_\eta\}$ of $\mathcal N$ such that every $O_\eta$ is diffeomorphic to open bounded subset of $\R^d$ and the bundle ${\mathcal W}$ is trivial on $O_\eta$. Let $\xi_\eta$ be a partition of unity $\xi_\eta$ supported on $O_\eta$, $\sum_\eta \xi_\eta(x)=1$. We constructed $H_\eta$ on every $O_\eta$. We can take $H(x,\lambda)=\sum_\eta \xi_\eta(x) H_\eta(x,\lambda)$.
\end{proof}

\bibliography{hadamard}{}
\bibliographystyle{ieeetr}

\end{document}